\documentclass[runningheads]{llncs}
\usepackage[normalem]{ulem}
\usepackage[T1]{fontenc}
\usepackage{adjustbox}
\usepackage{multirow}
\usepackage{graphicx}
\usepackage{xcolor}
\usepackage{amsmath,amssymb,amsfonts}

\usepackage{hyperref}

\usepackage[linesnumbered, noend]{algorithm2e}
\SetKwInOut{Input}{input}
\SetKwInOut{Output}{output}
\RestyleAlgo{ruled}

\newcommand{\BWT}{\ensuremath{\mathrm{BWT}}}
\newcommand{\SA}{\ensuremath{\mathrm{SA}}}
\newcommand{\LCP}{\ensuremath{\mathrm{LCP}}}
\newcommand{\BWTr}{\ensuremath{\mathrm{BWT}_r}}
\newcommand{\SAr}{\ensuremath{\mathrm{SA}_r}}
\newcommand{\LCPr}{\ensuremath{\mathrm{LCP}_r}}
\newcommand{\PSV}{\ensuremath{\mathrm{PSV}}}
\newcommand{\NSV}{\ensuremath{\mathrm{NSV}}}
\newcommand{\LF}{\ensuremath{\mathrm{LF}}}

\newcommand{\SSS}{\mathcal S}

\newcommand{\rev}[1]{#1^{\text{rev}}}

\begin{document}
\title{Practical Linear-Time Computation \\ of Smallest Suffixient Sets\thanks{Funded by Basal Funds FB0001 and AFB240001, and Fondecyt grant 1260080, ANID, Chile.}}
\author{
Francisco Olivares\inst{1,2}\orcidID{0000-0001-7881-9794} \and
Gonzalo Navarro\inst{1,2}\orcidID{0000-0002-2286-741X}
}
\authorrunning{F. Olivares and G. Navarro}
\institute{CeBiB -- Centre for Biotechnology and Bioengineering, Chile \and 
Department of Computer Science, University of Chile, Chile\\
\email{\{folivares,gnavarro\}@uchile.cl}
}
\maketitle              % typeset the header of the contribution
\begin{abstract}
Suffixient arrays are recent structures that have attracted attention because they offer relevant pattern matching functionality in less asymptotic space than the Run-Length BWT, the de-facto standard to index highly repetitive string collections. Various algorithms exist for building them from the suffix array data structures. We present the first construction algorithm that is (i) linear-time, (ii) one-pass over the structures, and (iii) implemented and practical. This makes the construction particularly useful on large text collections, which we demonstrate empirically by showing that it dominates the space/time tradeoff map of the implemented constructions.
\end{abstract}

\section{Introduction and Related Work}

Many of the fastest-growing text collections are highly repetitive, meaning that most documents are near-copies of others: genome collections like the One Million Genome Initiative, software repositories like Software Heritage, and versioned text collections like Wikipedia are examples of this phenomenon. Effectively searching such huge collections requires using storage formats that exploit repetitiveness to reduce space while offering fast indexed searches. The Run-Length Burrows-Wheeler Transform (RLBWT) is at the root of  many successful compressed indexes for this scenario \cite{BW94,MNSV09,GagieNP20,CGNtalg25}. The space of this family of indexes is dominated by $r$, the number of equal-letter runs in the BWT of the collection.

A very recent development \cite{Depuydt24} further pushes towards space reduction by defining a new measure $\chi \le r$, the size of the smallest so-called {\em suffixient set} of the text. This is a subset of the text positions that, if sorted by the corresponding prefixes in colexicographic order (i.e., comparing right-to-left), yields the so-called {\em suffixient array} \cite{TOCSpaper}. Provided with a method to access the compressed text, the suffixient array can efficiently search for one occurrence of a pattern in the text, or find all its maximal substrings that occur in the text. This is less than the classic pattern matching functionality of finding all the pattern occurrences, but sufficient for various applications while using less space.

Since smallest suffixient sets are meant to be used on large repetitive text collections, their efficient construction is key to their adoption. Cenzato et al.~\cite{CenzatoOP24,TOCSpaper} already proposed two efficient constructions, which work on top of the suffix array (SA), the Burrows-Wheeler Transform (BWT), and the longest common prefix array (LCP) of the reversed text. Their first algorithm \cite{CenzatoOP24}, performs one single pass over those structures (i.e., it is ``one-pass'') and takes time $O(n+\overline{r}\log\sigma)$, where $n$ is the length of the text, $\overline{r}$ is the $r$ value for the reverse text, and $\sigma$ is the alphabet size. Their second algorithm \cite{CenzatoOP24}, which we call LF, is linear-time but not one-pass, just as their third algorithm \cite{TOCSpaper}, which we call FM. FM is faster than LF but uses more space. Bonizzoni et al.~\cite{BonizzoniGR2026} recently proposed the first algorithm that is both linear-time and one-pass, which we call BGR, but no implementation is reported. Fujimaru et al.~\cite{fujimaru2026} present linear-time online algorithms that work directly on the text, processing it left-to-right or right-to-left. Fujimaru et al.~\cite{Fuji26} present another linear-time algorithm, which is however not one-pass nor implemented; it uses some structures built on the text and some on its reverse. They also present the first sublinear-time construction (also not implemented). Finally, K\"oppl and Kucherov \cite{koppl2026} gave a near-real time construction algorithm, without reporting an implementation. 

In this context, we present the first algorithm to build smallest suffixient sets that is (i) linear-time, (ii) one-pass, and (iii) implemented. See Table~\ref{tab:compar}. We empirically demonstrate its practicality by comparing it with the implemented alternatives: our algorithm outperforms previous algorithms, or matches the best ones, in both time and space. On passing, we implement BGR algorithm, which turns out to be practical but dominated by ours in both space and time.

\begin{table}[t]
\begin{center}
\begin{tabular}{l|c|c|c}
Algorithm & linear-time & one-pass & implemented \\
\hline
One-pass \cite{CenzatoOP24} & no & yes & yes \\
LF \cite{CenzatoOP24} & yes & no & yes \\
FM \cite{TOCSpaper} & yes & no & yes \\
BGR~\cite{BonizzoniGR2026} & yes & yes & \sout{no} yes \\
Online \cite{fujimaru2026} & yes & yes*$\!\!\!$ & yes \\
Linear \cite{Fuji26} & yes & no & no \\
Near-linear \cite{koppl2026} & no & yes*$\!\!\!$ & no \\
{\bf Ours} & yes & yes & yes \\
\hline
\end{tabular}
\end{center}
\caption{The existing algorithms to build a smallest suffixient set, and ours in this context. We implement algorithm BGR in this paper. Algorithms ``Online'' and ``Near-linear'' are one-pass over the text, not in the same sense as the others.}
\label{tab:compar}
\end{table}

\section{Preliminaries}\label{sec:preliminaries}
By $[i, j]$ we refer to all values $i, i + 1, \dots, j$, if $i \leq j$, and $[i, j] = \emptyset$, otherwise; and we define $[n] = [1, n]$. A string $T[1, n] \in \Sigma^n$ is a concatenation of $n$ elements of an alphabet $\Sigma = \{1, \dots, \sigma\}$. The length of $T$ is denoted by $\vert T \vert = n$. The empty string $\epsilon$ is the only string satisfying $\vert \epsilon \vert = 0$. $T[i]$ indicates the $i$-th element of $T$, and $T[i, j] = T[i] \dots T[j]$ is the substring of $T$ starting at position $i$ and ending at position $j$, if $1 \leq i \leq j \leq n$, and $T[i, j] = \epsilon$ otherwise. A prefix of $T$ is a substring of the form $T[1, j]$ and a suffix is a substring of the form $T[i, n]$. $\rev{T}$ denotes the reversal of $T$, that is, $\rev{T} = T[n] \dots T[1]$.

Given two strings $T, S \in \Sigma^+$, none a prefix of the other, and a total order over $\Sigma$, the \textit{lexicographic} order $<_{lex}$ over $\Sigma^+$ is defined as $T <_{lex} S$ if and only if $T[1,i] = S[1,i]$ and 
% esto complicaba las cosas si $S$ también se acababa $\vert T \vert = i$ or 
$T[i + 1] < S[i + 1]$, for some $i$.
The \textit{suffix array} of $T[1, n]$ ($\SA(T)$) is a permutation of the elements of $T$ such that $T[\SA(T)[i], n] <_{lex} T[\SA(T)[j], n]$ for any $1 \leq i < j \leq n$, namely $\SA(T)[i]$ is the $i$-th suffix of $T$ in lexicographic order. We assume $T$ ends with a special character $\$$, which is smaller than any other symbol in $\Sigma$ and only occurs in $T[n]$, so $\SA$ is well defined. The \textit{longest common prefix} array of $T$ ($\LCP(T)$) is defined as $\LCP(T)[i] = \ell$ if and only if $2 \leq i$ and $T[\SA(T)[i - 1], \SA(T)[i - 1] + \ell - 1] = T[\SA(T)[i], \SA(T)[i] + \ell - 1]$ and $T[\SA(T)[i - 1] + \ell] \neq T[\SA(T)[i] + \ell]$, that is, $\LCP(T)[i]$ is the length of the maximal prefix shared by the $(i - 1)$-th and the $i$-th suffixes of $T$ in the lexicographic order; we further define $\LCP(T)[1] = 0$. The \textit{Burrows-Wheeler Transform} array of $T$ ($\BWT(T)$) is defined as $\BWT(T)[i] = T[\SA(T)[i - 1]]$, if $2 \leq i$, and $\BWT(T)[1] = T[n]$, that is $\BWT(T)[i]$ is the symbol preceding the $i$-th sufix of $T$. By $\SAr$, $\LCPr$, and $\BWTr$ we refer to $\SA(\rev{T})$, $\LCP(\rev{T})$, and $\BWT(\rev{T})$, respectively, where in this case we assume $\rev{T} = \rev{(T[1,n-1])}\$$.

Given an integer array $\mathcal{I}$, the \textit{previous smaller value} array of $\mathcal{I}$ ($\PSV(\mathcal{I})$) is a lenght-$\vert \mathcal{I} \vert$ array defined as $\PSV(\mathcal{I})[i] = j$ if and only if $j = \max (\{k<i \mid \mathcal{I}[k] < \mathcal{I}[i] \} \cup \{0\})$, that is, $j$ is the largest index in $[1, i - 1]$ for which $\mathcal{I}$ is strictly smaller than $\mathcal{I}[i]$, if such value exists, and $j = 0$ otherwise. The \textit{next smaller value} array, $\NSV(\mathcal{I})$, is defined analogously for the smaller value with an index bigger than $i$ (or $|\mathcal{I}|+1$ if none exists). For brevity, by $\PSV$ and $\NSV$ we will refer to $\PSV(\LCPr)$ and $\NSV(\LCPr)$, respectively.

A run in $\BWTr$ is a maximal contiguous substring that repeats a single symbol.
Given $i \in [2, n]$ and $a, c \in \Sigma$, with $a \neq c$, we say that $i$ is a $c$-run break if $\BWTr[i - 1, i] = ac$ or $\BWTr[i - 1, i] = ca$. 
By $box(i) = [l_i, r_i]$ we refer to the maximal interval such that $i \in [l_i, r_i]$ and $\LCPr[i] \leq \LCPr[j]$, for any $j \in [l_i, r_i]$, in other words, $box(i) = [\PSV[i] + 1, \NSV[i] - 1]$. By $text(i) = n - \SAr[i] + 1$ we indicate the position in $T$ corresponding to $\BWTr[i]$. Similarly, by $bwt(j) = \SAr^{-1}[n + 1 - j]$ we denote the position in $\SAr$ of the symbol $T[j]$.

We say the substring $T[i, j]$ ($i - 1 \leq j$) is right-maximal if $j = n$ (so $T[i, j]$ is a suffix of $T$) or there exist $a, b \in \Sigma$ such that $a \neq b$ and both $T[i, j] \cdot a$ and $T[i, j] \cdot b$ occur in $T$. Note that, by definition, $\epsilon$ is always right-maximal.

\subsection{Suffixient sets}\label{sec:suffixient}

Suffixient sets were initially proposed by Depuydt et al.~\cite{Depuydt24}.

\begin{definition}[Suffixient set \cite{TOCSpaper,CenzatoOP24,Depuydt24}]\label{def:suffixient}
A set $S \subseteq [n]$ is  \emph{suffixient} for a string $T$ if, for every one-character right-extension $T[i,j]$ ($j\geq i$) of every right-maximal string $T[i,j-1]$, there exists $x\in S$ such that  $T[i,j]$ is a suffix of $T[1,x]$.
\end{definition}

They also proposed a compressed index based on a suffixient set of a text $T$ that can efficiently find one occurrence of any given pattern $P$ within $T$, or say that $P$ does not occur in $T$. The index can also find all the maximal exact matches (MEMs) of $P$. They show how to compute a suffixient set of size $O(\overline{r})$, where $\overline{r}$ is the number of runs in $\BWTr$.

Their index uses $O(\vert S \vert + g)$ words of space, where $g$ is the size of a grammar that gives random access to the input text and $S$ is the suffixient set they build. It follows that the smaller the suffixient set built, the smaller their index is.

Cenzato et al.~\cite{CenzatoOP24} followed up and studied methods to build suffixient sets of smallest cardinality. They introduced three algorithms to achieve this goal. The first of them (Algorithm~1) runs in quadratic-time and is conceptually simple. The second one (Algorithm~2) runs in time $O(n + \overline{r}\log\sigma)$. The algorithm is called \textit{one-pass} in the paper because it needs just one sequential pass over $\SAr$, $\LCPr$, and $\BWTr$, which enables \textit{prefix free parsing}, a technique which allows streaming those arrays in compressed space \cite{BoucherGKLMM19}. The third algorithm (Algorithm~4) runs in $O(n)$ time and uses the same three arrays. Since it relies on the $\LF$ mapping \cite{FerraginaM05}, it needs random access on $\LCPr$, so it does not run on top of \textrm{PFP}. Later \cite{TOCSpaper}, the same authors propose another linear-time algorithm (Algorithm 7), which uses $\PSV$/$\NSV$ arrays and, in practice, is the fastest between the four algorithms, at the cost of using more working space. In the same paper, the authors proposed an index called \textit{suffixient array}, a sampling of the positions of the \textit{prefix array} (the mirroring version of $\SA$, which uses lexicographic order over the reversed suffixes) which can be binary-searched and offers fast pattern matching, if random access to the text is still available. This index outperforms many well-known ones, such as the $r$-index \cite{GagieNP20}, the fastest index in practice offering full pattern matching (i.e., finding all the pattern occurrences, not just one).

Several other construction algorithms were proposed after the work of Cenzato et al.~\cite{CenzatoOP24,TOCSpaper}. Bonizzoni et al.~\cite{BonizzoniGR2026} proposed the first algorithm that is linear-time and one-pass. Fujimaru et al.~\cite{fujimaru2026} proposed an online algorithm that works directly on the text (so it avoids computing $\SAr$, $\LCPr$, and $\BWTr$), and later \cite{Fuji26} presented another linear-time algorithm that builds some structures on the text and on its reverse, but it is not linear-time nor one-pass. They also present the first sublinear-time construction, which relies on some complex data structures proposed by Kempa and Kociumaka \cite{sublinear_suffix_array_construction}. Finally, K\"oppl and Kucherov \cite{koppl2026} gave the first near-real time construction algorithm.

Suffixient sets have also become a field of study from other perspectives. Cenzato et al.\cite{cenzato2025testingsuffixientsets} showed how to determine in linear time if an arbitrary set is suffixient of smallest cardinality. The measure $\chi$ (the size of a smallest suffixient set) is also studied as a measure of repetitiveness, due to the combinatorial properties of suffixient sets \cite{NRUspire25.2}. It is known that $\chi = O(\overline{r})$, and also that $\chi = O(r)$ \cite{NRUspire25.2}, so $\chi = O(\min(\overline{r}, r))$.

\section{A Simple Linear-Time Construction Algorithm}

Cenzato et al.~\cite[Alg.~3]{TOCSpaper} give a simple algorithm to compute a suffixient set of smallest cardinality. They define a set $B_{i, c}$ that contains the positions in $box(i)$ that are $c$-run breaks in $\BWTr$. Then, for a $c$-run break $i$, with $\BWTr[i'] = c$ for some $i' \in \{i - 1, i\}$, the algorithm relies on adding to the resulting set the position $text(i')$ if and only if it holds \begin{equation}i = \max\{j\ :\ j\in B_{i, c} \wedge \LCPr[j] = \max \LCPr[B_{i, c}]\}.\label{eq:rightmost-maximum}\end{equation} That algorithm spends $O(n)$ time to determine if Eq.~(\ref{eq:rightmost-maximum}) holds, and, since boxes may overlap, it spends quadratic time to run the whole algorithm. In this section we present a method to evaluate Eq.~(\ref{eq:rightmost-maximum}) in $O(1)$ time, which leads to a linear-time version of the algorithm. As a first step, we introduce the notion of \textit{last c-candidate} positions (Definition~\ref{def:last-candidate}). Later, we prove that \textit{last c-candidate} positions are exactly the positions that satisfy Eq.~\ref{eq:rightmost-maximum} (Proposition~\ref{prop:last-candidate}). Finally, we show a modified version of the quadratic algorithm using this equivalence (Algorithm~\ref{alg:lc-algo}).

In the rest of this section, we assume there are exactly $k$ $c$-run breaks in $\BWTr$. Let $1 \leq i_1 < i_2 < \dots < i_k \leq n$ be the positions of those run-breaks, and let $box(i_h) = [l_h, r_h]$, for each $h \in [k]$.

\begin{definition}[\textit{last c-candidate}]\label{def:last-candidate}
Let $i_h$ be the position of a $c$-run break. We say $i_h$ is a \textit{c-candidate} if one of the following conditions holds: 

\begin{itemize}
    \item{$h = 1$,}
    \item{$i_{h - 1} < l_h$, or}
    \item{$i_h \le r_{h - 1}$ and $i_{h - 1}$ is a $c$-\textit{candidate}.}
\end{itemize}

Moreover, we say that $i_h$ is a \textit{last c-candidate} if it is a \textit{c-candidate} and one of the following conditions hold:

\begin{itemize}
    \item{$h = k$, or}
    \item{$r_h < i_{h + 1}$.}
\end{itemize}
\end{definition}

Intuitively, a position is a \textit{c-candidate} if it is a $c$-run break and its $\LCPr$ value is equal to the $\LCPr$ value of every previous $c$-run break (if any) within its box. On the other hand, a \textit{c-candidate} position is a \textit{last c-candidate} if there are no larger \textit{c-candidates} positions within its box. The following proposition establishes the equivalence between finding $c$-run break positions that are the rightmost maximum within a box and \textit{last c-candidate} positions.

\begin{proposition}\label{prop:last-candidate}
Let $i_h$ be a $c$-run break. We have $$i_h = \max\{i\ :\ i\in B_{i_h, c} \wedge \LCPr[i] = \max \LCPr[B_{i_h, c}]\}$$ if and only if $i_h$ is a \textit{last c-candidate}.
\end{proposition}
\begin{proof}
Let $B_{i_h, c} = \{i_{h - a}, \dots, i_h, \dots, i_{h + b}\}$, $box(i_h) = [l_h, r_h]$, and define $i_0 := 0$ and $i_{k + 1} := n + 1$.
By definition of $B_{i_h, c}$, for every $j \in [h - a, h + b]$ it holds $\LCPr[i_h] \leq \LCPr[i_j]$. In addition, it holds 
\begin{equation}\label{eq:lc-boxes}i_{h - a - 1}  < l_h \leq i_{h - a} < \dots < i_h < \dots < i_{h + b} \leq r_h < i_{h + b + 1}
\end{equation} 

\noindent($\Rightarrow$)
Suppose $i_h = \max\{i\ :\ i \in B_{i_h, c} \wedge \LCPr[i] = \max \LCPr[B_{i_h, c}]\}$. By the \textit{maximality} of $\LCPr[i_h]$, for every $i_j \in B_{i_h, c}$ it holds $\LCPr[i_h] = \LCPr[i_j]$. This implies $b = 0$ and $box(i_{h - a}) = \dots = box(i_h)$. Since by Eq.~(\ref{eq:lc-boxes}) it holds $i_{h - a - 1} < l_{h - a} \leq i_{h - a}$, $i_{h - a}$ is a \textit{c-candidate}. In addition, by Eq.~(\ref{eq:lc-boxes}), for $d = 1, \dots, a$  we have $i_{h - a + d} < r_{h - a + d - 1}$, and by the equality of the boxes we have $r_{h - a + d - 1} = r_h$, then $i_{h - a + d}$ is a \textit{c-candidate}, so $i_h$ is a \textit{c-candidate}. Finally, by Eq.~(\ref{eq:lc-boxes}) it holds $r_h < i_{h + 1}$, therefore, $i_h$ is a \textit{last c-candidate}.

\noindent($\Leftarrow$) We proceed by contrapositive. Suppose $i_h \neq i_m = \max\{i\ :\ i\in B_{i_h, c} \wedge \LCPr[i] = \max \LCPr[B_{i_h, c}]\}$. If $m < h$, since $i_m \in B_{i_h, c}$, there exists $e \in [h - m]$ such that $\LCPr[i_{h - e}] > \LCPr[i_{h - e + 1}] = \dots = \LCPr[i_{h}]$. Because $\LCPr[i_{h - e }] > \LCPr[i_{h - e + 1}]$, we have $l_{h - e + 1} \le i_{h - e}$ and $r_{h - e} < i_{h - e + 1}$, so $i_{h - e + 1}$ is not a \textit{c-candidate}. Now, for $d = 1, \dots, e$ we have $l_{h - e + d} \le i_{h - e + d - 1}$ and $i_{h - e + d - 1}$ is not a \textit{c-candidate}, then $i_{h - e + d}$ is not a \textit{c-candidate}, which implies $i_h$ is not a \textit{c-candidate}. On the other hand, if $h < m \leq k$, then $1 \leq b$ and $i_{h + 1} \in B_{i_h, c}$, so we have $\LCPr[i_h] \leq \LCPr[i_{h + 1}]$, which means $i_{h + 1} \le r_h$. Therefore, $i_h$ cannot be a \textit{last c-candidate}.\qed
\end{proof}

From Proposition~\ref{prop:last-candidate} we can develop a particular strategy to find local maxima inside \textit{boxes} by finding \textit{last candidate}s within them. Actually, since the \textit{candidate} condition propagates forward, a single sequential scan of the $\BWTr$ suffices to identify \textit{last candidate} positions, given that we know boxes' boundaries. As mentioned in Section~\ref{sec:preliminaries}, we can compute $box(i) = [l_i, r_i] = [\PSV[i] + 1, \NSV[i] - 1]$ in $O(1)$ time. 

Following the above ideas, we show a \textit{linearization} of the quadratic algorithm in Algorithm~\ref{alg:plain-lc-algo}. We call this algorithm \textit{plain last-candidate} (\textrm{PLC}) algorithm (the reason for ``plain'' will be made clear later). It uses a data structure $R$ to save information about the \textit{candidate} condition on the last $c$-run break seen up to now. The algorithm sequentially scans the $\BWTr$ looking for run-breaks. For a given $c$-run break $i$, where $\BWTr[i'] = c$ for some $i' \in \{i -1, i\}$, $R[c]$ stores $(sa\_pos \gets i, text\_pos \gets n - \SAr[i'] + 1, candidate \in \{true, false\})$, where $candidate$ is set to $true$ if and only if position $i$ is a \textit{c-candidate}.

Algorithm~\ref{alg:plain-lc-algo} shows the \textrm{PLC} algorithm.
We next prove its correctness. 

\begin{algorithm}[t]
% LC stands for "Last Candidate". Very unoriginal
\caption{PLC-Algorithm}\label{alg:plain-lc-algo}
\Input{A string $T[1,n]$ over a finite alphabet $\Sigma$.}
\Output{A smallest suffixient set for $T$.}
{$\mathcal{S} \gets \emptyset$}\;
{$\BWTr\ \gets \textrm{BWT}(\rev{T})$;
$\LCPr\ \gets \textrm{LCP}(\rev{T})$;
$\SAr\ \gets \textrm{SA}(\rev{T})$\label{line:plc.compute arrays}\;
$\NSV \gets \textrm{NSV}(\textrm{LCP}) $;
$\PSV \gets \textrm{PSV}(\textrm{LCP}) \label{line:plc.compute-sv}$\;
}\label{line:lc.compute arrays sv}
$R[1,\sigma] \gets (sa\_pos \gets 1, text\_pos \gets 0, candidate \gets true) \times \sigma$\;\label{line:plc.init}
\For{$i = 2, \dots, n$}{
\If{$\BWTr[i] \neq \BWTr[i - 1]$}{
\For{$i'\in \{i - 1, i\}$}{
$candidate \gets false $; $c \gets \BWTr[i']$\;\label{line:plc non candidate}
\If{$\NSV[R[c].sa\_pos] \leq i \land R[c].candidate = true$ \label{line:plc.check lc}}{
$\SSS \gets \SSS \cup \{R[c].text\_pos\}$\;
}
\Else{
$candidate \gets R[c].candidate$\;
}
\If{$R[c].sa\_pos \leq \PSV[i]$ \label{line:plc.check candidate}}{
$candidate \gets true$\;\label{line:plc.candidate true}
}
$R[c] \gets (i, n - \SAr[i'] + 1, candidate)$\;\label{line:plc.update}
}}}
\ForEach{$c \in \Sigma$\label{line:plc.eval}}
{
\lIf{$R[c].candidate = true$}{$\SSS \gets \SSS \cup \{R[c].text\_pos\}$}
}
\Return $\SSS$\;
\end{algorithm}

\begin{lemma}
Given $T \in \Sigma^n$, Algorithm~\ref{alg:plain-lc-algo} computes a suffixient set of smallest cardinality $\SSS$ in $O(n)$ time and $O(n)$ space.
\end{lemma}
\begin{proof}
For each $c \in \Sigma$, $R[c]$ is set to $(1, 0, true)$ at the beginning of the algorithm (line~\ref{line:plc.init}). Then, the first time we find a $c$-run break, the condition of line~\ref{line:plc.check lc} is never satisfied (because for every $c$-run-break we initially set $R[c].sa\_pos = 1$ and $\NSV[1] = n + 1$), but the condition of line~\ref{line:plc.check candidate} is true, so $R[c]$ is updated and set as a \textit{c-candidate} (because the first $c$-run break is always a \textit{c-candidate}).
Now, assume that we have scanned the $\BWTr$ until position $i - 1$ and $i$ is a $c$-run break, with $\BWTr[i'] = c$ for some $i' \in \{i - 1, i\}$. Also, assume that $j < i$ is the previous $c$-run break, $j' \in \{j - 1, j\}$ and $\BWTr[j'] = c$. Then $R[c].sa\_pos = j$, $R[c].text\_pos = n - \SAr[j'] + 1$. If $R[c].candidate = true$ and $\NSV[R[c].sa\_pos] \leq i$ (line~\ref{line:plc.check lc}), then $j$ is a \textit{last c-candidate} and we add $R[c].text\_pos$ to the suffixient set we have built up to now. On the other hand, if $R[c].sa\_pos \leq \PSV[i]$ or $j$ is a \textit{c-candidate} and $i < \NSV[R[c].sa\_pos]$ (line~\ref{line:plc.check candidate}) we set $i$ as a \textit{c-candidate} (lines~\ref{line:plc.candidate true} and \ref{line:plc.update}), otherwise it is set as a non-$\textit{c-candidate}$ (lines~\ref{line:plc non candidate} and \ref{line:plc.update}). Now, we update the remaining attributes of $R[c]$ according to $i$'s values: $R[c].sa\_pos = i$, $R[c].text\_pos = n - \SAr[i'] + 1$. Finally, if $i$ is the last $c$-run break, it is added to the output suffixient set in the final loop if it is a \textit{c-candidate} (line \ref{line:plc.eval}). Therefore, at the end of the algorithm execution, the output set $\SSS$ consists of all (and only the) \textit{last candidate} positions. Then, Algorithm~\ref{alg:plain-lc-algo} correctness follows from correctness of the quadratic algorithm \cite{CenzatoOP24,Depuydt24}, and Proposition \ref{prop:last-candidate}.

The arrays $\BWTr$, $\LCPr$, $\SAr$, $\NSV$, and $\PSV$ consume $O(n)$ space. The data structure $R$ uses $O(\sigma)$ space and the set $\SSS$ is also of size $O(n)$. Thus, the space usage of Algorithm~\ref{alg:plain-lc-algo} is $O(n)$. With regard to running time, computing $\BWTr$, $\LCPr$, $\SAr$  in line \ref{line:plc.compute arrays} takes $O(n)$ time. $\PSV$ and $\NSV$ in line \ref{line:plc.compute-sv} can also be computed in $O(n)$ time \cite{BerkmanSV93}. Scanning $\BWTr$ takes $O(n)$ time, and we perform $O(1)$ operations for each $\BWTr$ position. The total running time is then bounded by $O(n)$.\qed
\end{proof}

Cenzato et al.~\cite{TOCSpaper} 
give two linear-time algorithms to compute a smallest suffixient set. One is based on the $\LF$ array (so we refer to it as $\LF$ algorithm) and the other is based on $\PSV/\NSV$ arrays (and we call it \textrm{FM} algorithm). As pointed out in the paper, the \textrm{FM} algorithm is faster than the \textrm{LF} algorithm, at the cost of increasing space usage. As shown in Table~\ref{fig:results}, the $\textrm{PLC}$ algorithm is neither faster than the \textrm{FM} algorithm nor uses less space. Interestingly, with a slight modification, we can turn it as fast as the \textrm{FM} algorithm using as little space as the \textrm{LF} algorithm, getting the best performance of both. We refer to the new algorithm as \textrm{LC}.

By the form in which Algorithm~\ref{alg:plain-lc-algo} works, we can see that it never asks for $\PSV$/$\NSV$ values at positions that are not run-breaks. So, there is no need to compute $\PSV$/$\NSV$ for other positions. Further, we can compute the needed values \textit{on the fly} as we find the run-breaks, using a stack-based strategy similar to the one of Bonizzoni et al.~\cite{BonizzoniGR2026}. In particular, we use two stacks, $psv$ and $nsv$. In $psv$, we push each $i$ as we iterate from $i = 2$ to $i = n$, that is, every $i$ can be the $\PSV$ value for some run-break. Then, for each run-break $i$ we pop elements from $psv$ until $\LCPr[psv.top()] < \LCPr[i]$ or $psv$ is empty, so $\PSV[i] = psv.top()$ or $\PSV[i] = 1$, respectively. Since we compute each needed $\PSV$ value when it is requested, we store it in a variable $psv_i$, so there is no need to create the array $\PSV$, which reduces running time and space usage. On the other hand, in $nsv$ we push a $c$-run break $j$ to mean it is still pending to find its $\NSV$ value, which will not be required until the next occurrence of a $c$-run break. At the beginning of iteration $i$, while $\LCPr[i] < \LCPr[nsv.top()]$ we set $R[\BWTr[nsv.top() - 1]].nsv = R[\BWTr[nsv.top()]].nsv = i$ (because we had defined $R[\BWTr[j - 1]].sa\_pos = R[\BWTr[j]].sa\_pos = j$ when we processed the run-break $j$) and call $nsv.pop()$ (because we already found $nsv.top()$ $\NSV$ value). Since we store $\NSV$ values only for the run-breaks, we can store those values in the data structure $R$, adding a new field $R[c].nsv$, which is set $R[c].nsv = n + 1$ as a default value when we update $R[c]$ on line \ref{line:plc.update} until we compute its real value (which can be also $n + 1$). Then there is also no need to create an array $\NSV$, which again reduces running time and space usage. Algorithm~\ref{alg:lc-algo} applies these modifications to Algorithm~\ref{alg:plain-lc-algo}.

\begin{algorithm}[t]
% LC stands for "Last Candidate".
\caption{LC-Algorithm}\label{alg:lc-algo}
\Input{A string $T[1,n]$ over a finite alphabet $\Sigma$.}
\Output{A smallest suffixient set for $T$.}
{$\mathcal{S} \gets \emptyset$}\;
{$\BWTr\ \gets \textrm{BWT}(\rev{T})$;
$\LCPr\ \gets \textrm{LCP}(\rev{T})$;
$\SAr\ \gets \textrm{SA}(\rev{T})$ \; 
}
$R[1,\sigma] \gets (sa\_pos \gets 1, text\_pos \gets 0, candidate \gets true, nsv \gets n + 1) \times \sigma$\;
$psv \gets new\_stack()$; $nsv \gets new\_stack()$; $psv.push(1)$; $nsv.push(1)$\;
\For{$i = 2, \dots, n$}{
\While{$nsv$ is not empty and $\LCPr[i] < \LCPr[nsv.top()]$}{
\lFor{$j \in \{nsv.top() - 1, nsv.top()\}$}{
$R[\BWTr[j]].nsv \gets i$}
$nsv.pop()$\;
}
\If{$\BWTr[i] \neq \BWTr[i - 1]$}{
$nsv.push(i)$\;
\lWhile{$psv$ is not empty and $\LCPr[i] \leq \LCPr[psv.top()]$}{$psv.pop()$}
\lIf{$psv$ is not empty}{$psv_i \gets psv.top()$ \textbf{else} $psv_i \gets 1$}
\For{$i'\in \{i - 1, i\}$}{
$candidate \gets false $; $c \gets \BWTr[i']$\;%\label{line:lc non candidate}
\If{$R[c].nsv \leq i \land R[c].candidate = true$}{
$\SSS \gets \SSS \cup \{R[c].text\_pos\}$\;
}
\Else{
$candidate \gets R[c].candidate$\;
}
\If{$R[c].sa\_pos \leq psv_i$}{
$candidate \gets true$\;
}
$R[c] \gets (i, n - \SAr[i'] + 1, candidate, n + 1)$\;
}}
$psv.push(i)$\;
}
\ForEach{$c \in \Sigma$}
{
\lIf{$R[c].candidate = true$}{$\SSS \gets \SSS \cup \{R[c].text\_pos\}$}
}
\Return $\SSS$\;
\end{algorithm}

\section{Experimental Results}

In this section we empirically compare our new algorithm with previous ones that are implemented. We also implement and compare BGR algorithm \cite{BonizzoniGR2026}.

The experiments were run using a workstation Intel(R) Pentium(R) IV, CPU @ 2.4 GHz with 8 cores, 10 MB cache, 256 gigabytes of RAM, 860 gigabytes of disk running Devuan GNU/LINUX 2.1. We compared the running time and the space consumption of algorithm \textrm{LC} against \textrm{PLC}, \textrm{FM}, and \textrm{LF}, computing smallest suffixient set over DNA sequences from the Pizza\&Chili collection\footnote{\url{https://pizzachili.dcc.uchile.cl}}. All the implementations of those algorithms are publicly available at {\tt \url{https://github.com/regindex/suffixient}}. We used the \textit{elapsed time} and the \textit{maximum resident set size} metrics returned by the utility \verb|usr/bin/time|\footnote{\url{https://www.gnu.org/software/time/}} as the running-time and the space consumption of the algorithms, respectively. We include in the construction cost the time to build the suffix array and other structures from the text. We also measure the performance of the algorithm of Bonizzoni et al. \cite{BonizzoniGR2026} implemented by us (called BGR, the implementation is available upon request) and the online algorithm \cite{fujimaru2026} (called Onl, also available upon request), the latter as a way to compare a different paradigm to compute suffixient sets. 

\begin{figure}[p]
    \centering
    \begin{minipage}{0.49\textwidth}
        \centering
        \includegraphics[width=\textwidth]{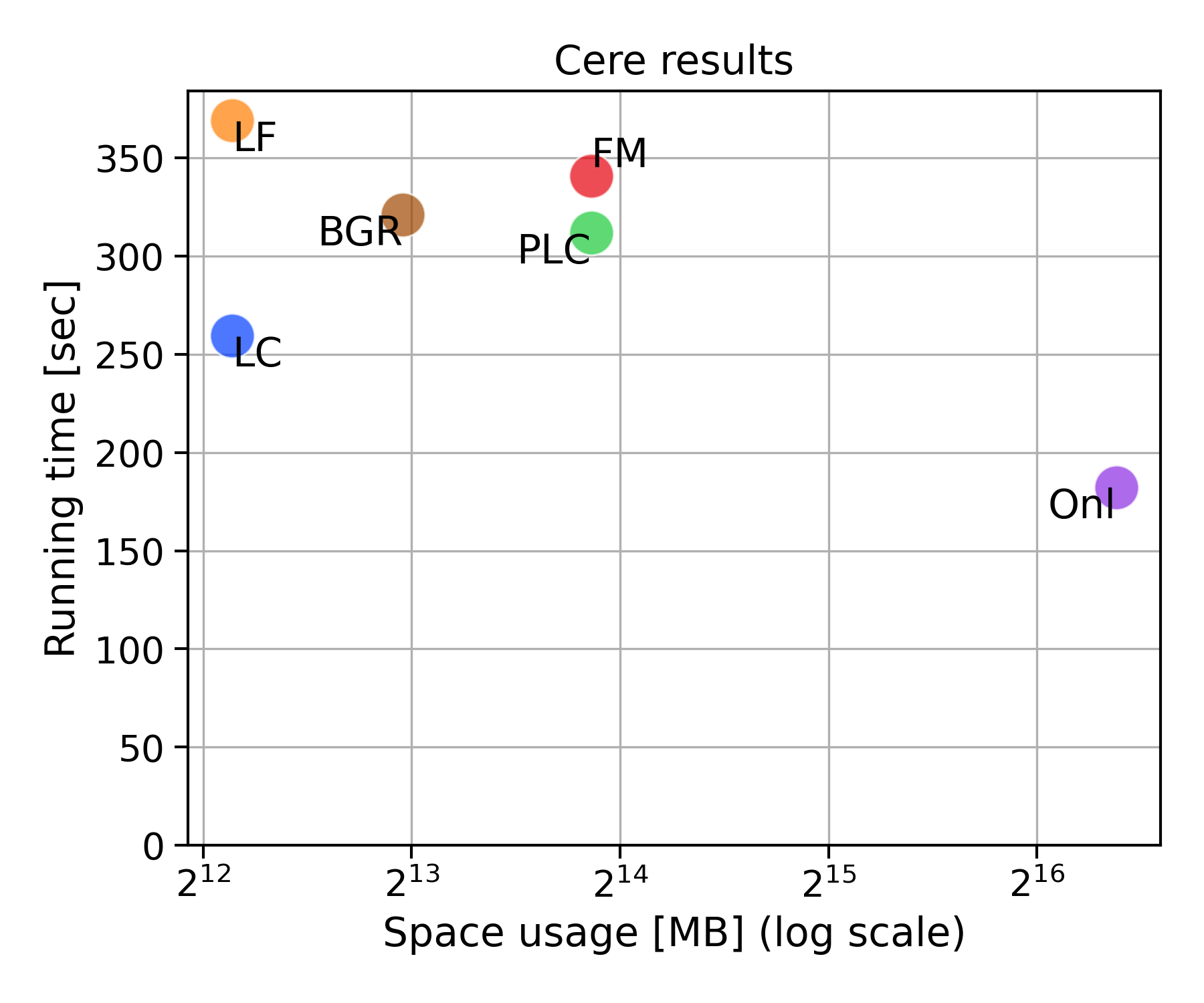} 
    \end{minipage}
    \begin{minipage}{0.49\textwidth}
        \centering
        \includegraphics[width=\textwidth]{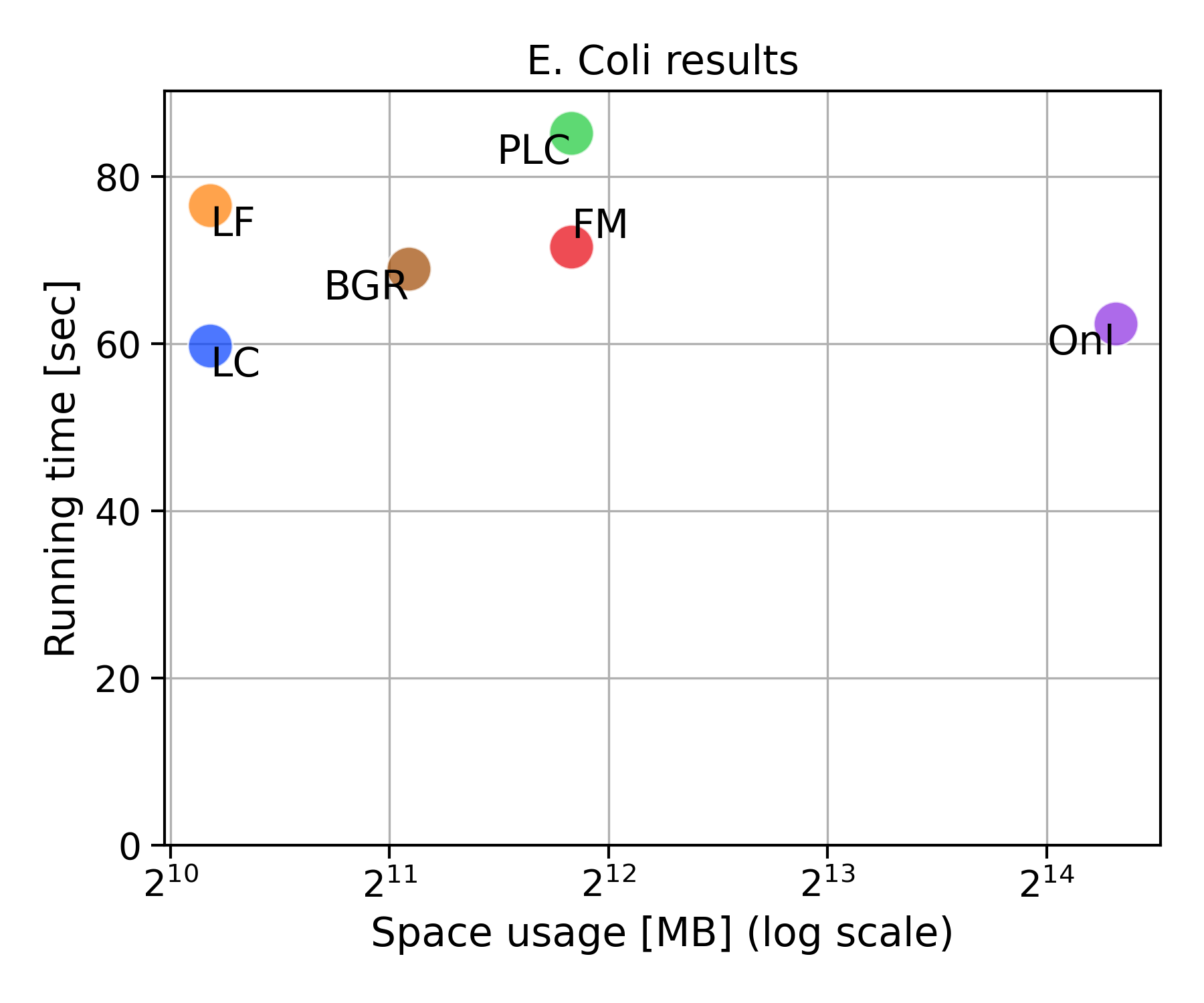} 
    \end{minipage}
    \begin{minipage}{0.49\textwidth}
        \centering
        \includegraphics[width=\textwidth]{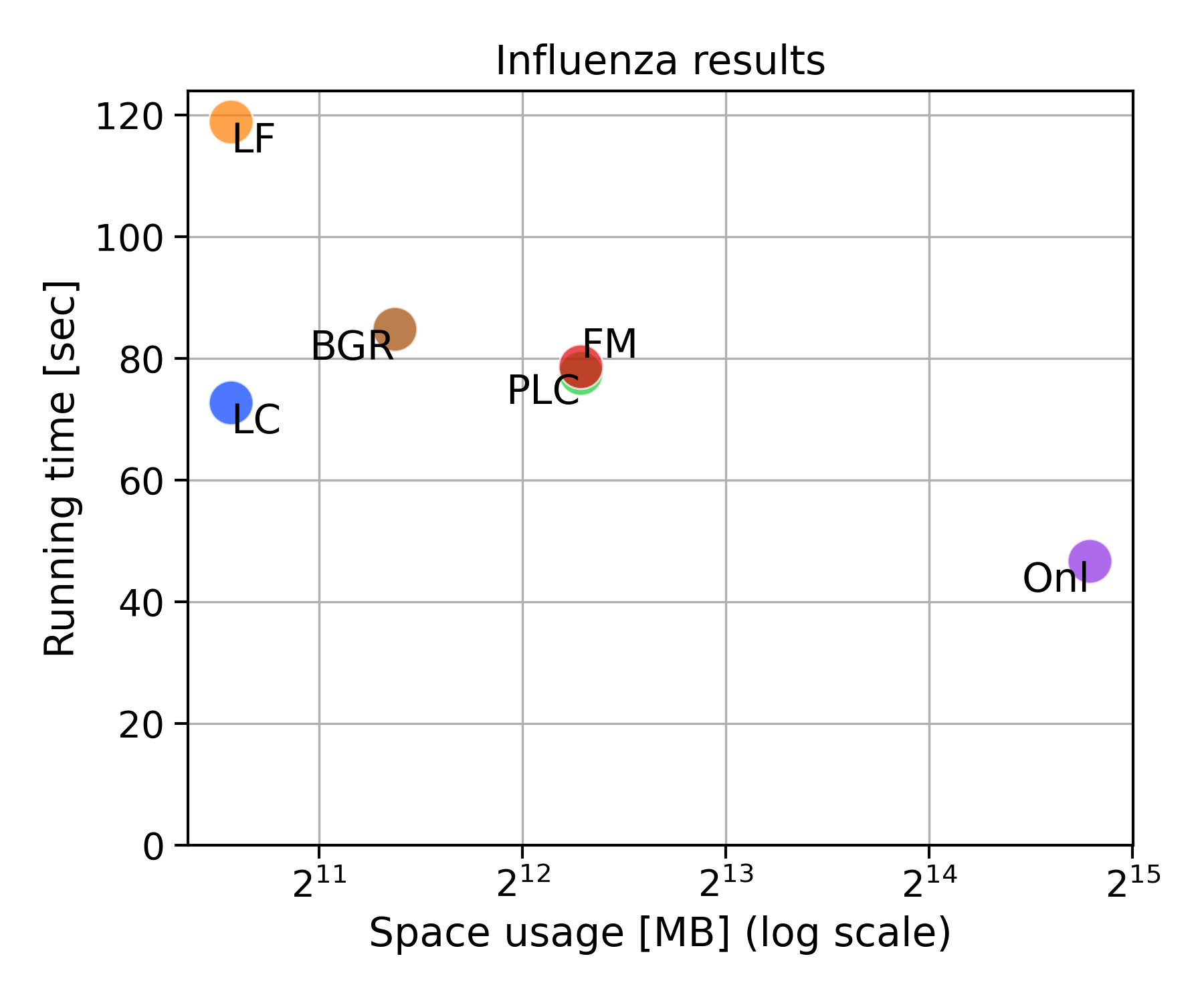} 
    \end{minipage}
    \begin{minipage}{0.49\textwidth}
        \centering
        \includegraphics[width=\textwidth]{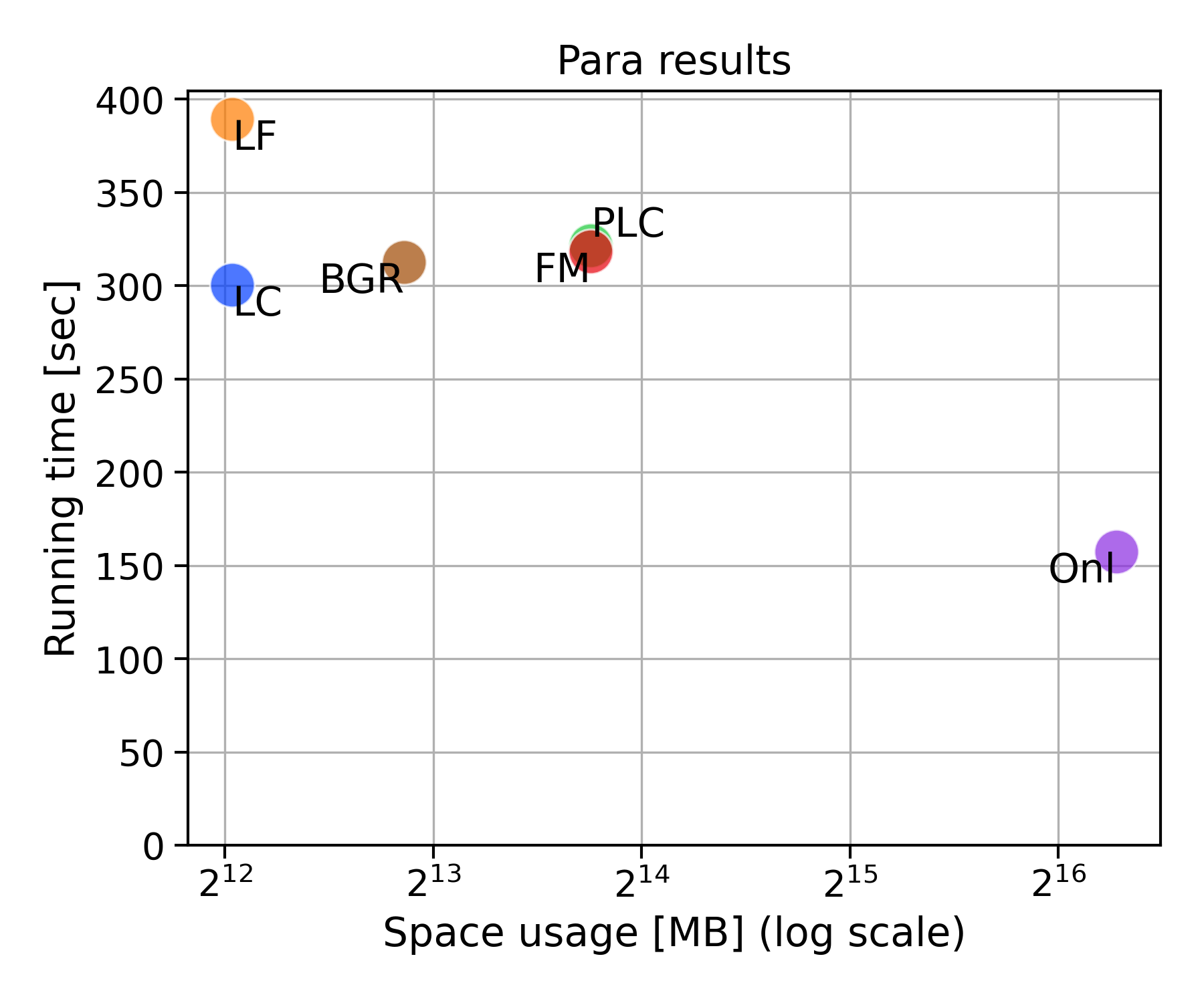} 
    \end{minipage}
    \begin{minipage}{0.49\textwidth}
        \centering
        \includegraphics[width=\textwidth]{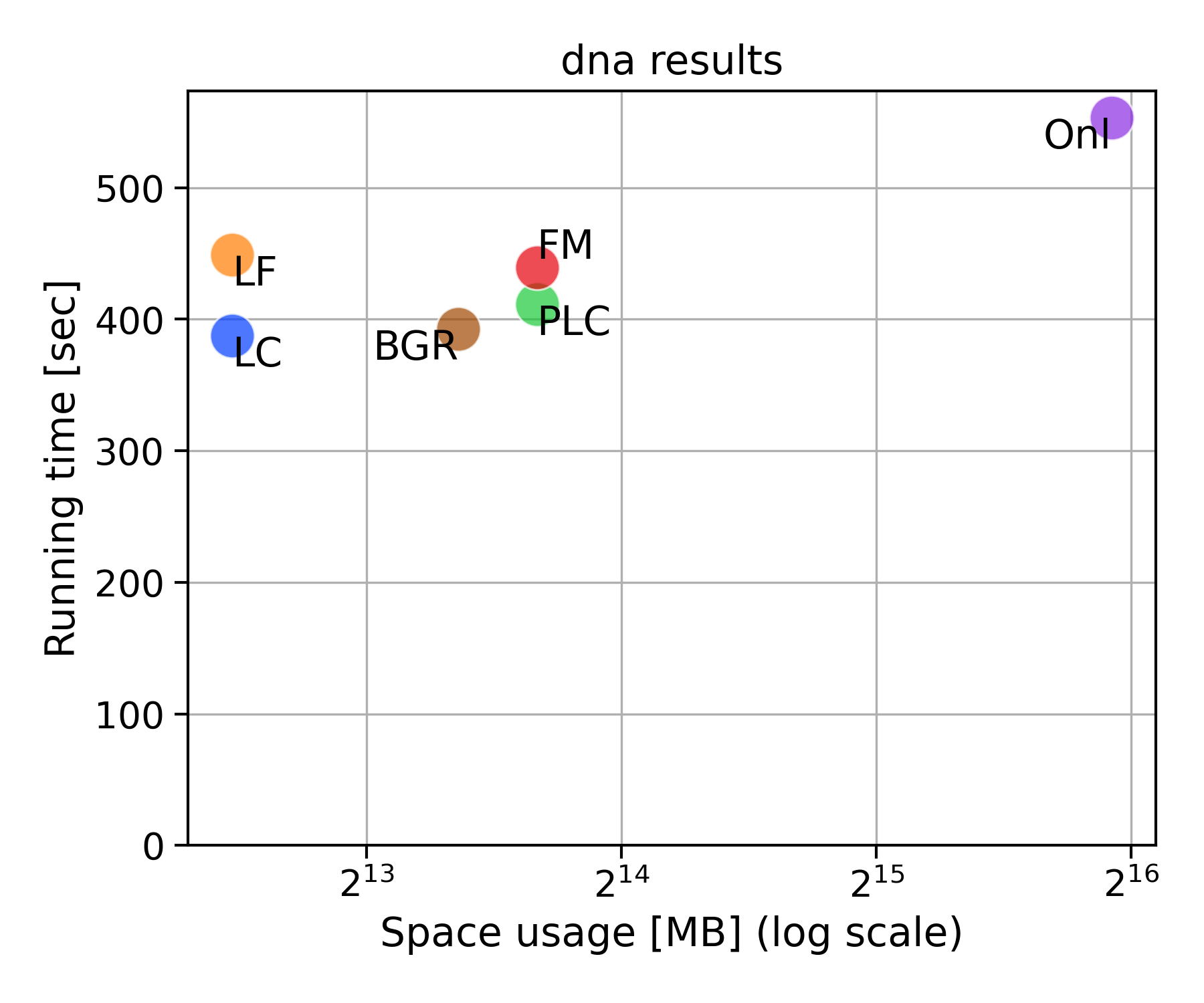}
    \end{minipage}
    \begin{minipage}{0.49\textwidth}
        \centering
        \includegraphics[width=\textwidth]{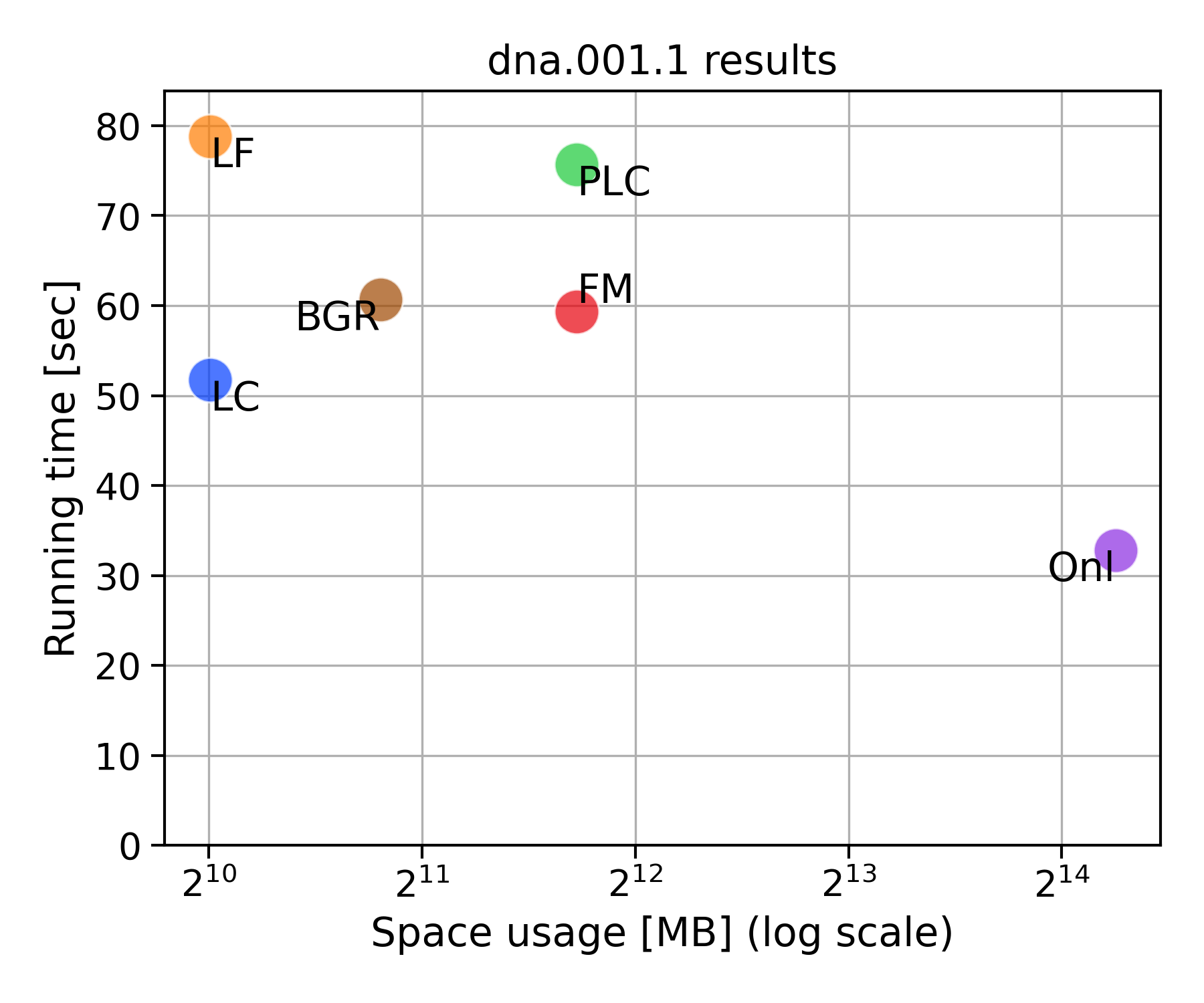} 
    \end{minipage}
    
    \bigskip
    
     \centering
 \begin{adjustbox}{width=\textwidth}
 \begin{tabular}{|l|r|r||r|r|r|r|r|r||r|r|r|r|r|r|} \hline
 \multicolumn{3}{|c||}{} & \multicolumn{6}{|c||}{\textbf{space (MB)}} & \multicolumn{6}{|c|}{\textbf{time (sec)}} \\ \hline
 Dataset $T$ & \multicolumn{1}{c|}{length $n$} & \multicolumn{1}{c||}{$\chi$} & \multicolumn{1}{|c|}{\textbf{LC}} & \multicolumn{1}{|c|}{\textbf{LF}} & \multicolumn{1}{|c|}{\textbf{PLC}} & \multicolumn{1}{|c|}{\textbf{FM}} & \multicolumn{1}{|c|}{\textbf{BGR}} & \multicolumn{1}{|c||}{\textbf{Onl}} & \multicolumn{1}{|c|}{\textbf{LC}} & \multicolumn{1}{|c|}{\textbf{LF}} & \multicolumn{1}{|c|}{\textbf{PLC}} & \multicolumn{1}{|c|}{\textbf{FM}} & \multicolumn{1}{|c|}{\textbf{BGR}} & \multicolumn{1}{|c|}{\textbf{Onl}} \\\hline
 Cere & 461,286,644 & 9,921,698 & $4509$ & $4509$  & $14872$ & $14873$ & $7954$ & $85336$ & $259.55$ & $368.95$ & $311.72$ & $340.68$ & $321.02$ & $182.15$ \\
 E.\ Coli & 112,689,515 & 13,119,749 & $1160$ & $1160$ & $3638$ & $3638$ & $2178$ & $20371$ & $59.74$ & $76.58$ & $85.20$ & $71.60$ & $68.96$ & $62.38$ \\
 Influenza & 154,808,555 & 2,225,543 & $1516$ & $1516$ & $4995$ & $4995$ & $2651$ & $28330$ & $72.73$ & $118.95$ & $77.65$ & $78.76$ & $84.81$ & $46.74$ \\
 Para & 429,265,758 & 13,382,184 & $4197$ & $4196$ & $13840$ & $13841$ & $7439$ & $79548$  & $300.61$ & $389.34$ & $321.62$ & $318.45$ & $312.66$ & $157.63$ \\
 dna & 403,927,746 & 215,193,300 & $5686$ & $5686$ & $13024$ & $13024$ & $10523$ & $62112$ & $387.57$ & $448.71$ & $411.25$ & $439.21$ & $392.48$ & $553.33$ \\
 dna.001.1 & 104,857,600 & 1,414,698 & $1028$ & $1028$ & $3386$ & $3386$ & $1790$ & $19500$ & $51.74$ & $78.84$ & $75.66$ & $59.32$ & $60.65$ & $32.78$  \\
% Einstein & 467,626,544 & 200,245 & $4571$ & $4571$ & $15077$ & $15077$ & $85854$ & $268.68$ & $359.12$ & $301.32$ & $310.3$ & $118.46$ \\
 \hline
 \end{tabular}
 \end{adjustbox}

    \caption{Space usage and running time comparison between all construction algorithms on DNA sequences from Pizza\&Chili collection. Logscale in the space coordinate.
    }
    \label{fig:results}
\end{figure}

The results on DNA are shown in Figure~\ref{fig:results}. We can see that \textrm{LC} used the same space as \textrm{LF}, 1.7--1.9 times less space than \textrm{BGR}, 3.1--3.3 times less space than \textrm{PLC} and \textrm{FM}, and 17--19 times less space than \textrm{Onl} (except on dna, where the ratios are 2.3 and 11, respectively). On the other hand, \textrm{LC} was 14\%--39\% faster than \textrm{LF}, 1\%--24\% faster than \textrm{BGR}, 6\%--32\% faster than \textrm{PLC}, and 6\%--24\% faster than \textrm{FM}. Compared to \textrm{Onl}, \textrm{LC} gave mixed results, being from 30\% faster to 48\% slower. From those results, we conclude that \textrm{LC} outperforms all the other algorithms in {\em both} time and space, with the exception of the running time of \textrm{Onl}, where the comparison is mixed. It must be noted that the good time performance of \textrm{Onl} comes at the price of using a large amount of space (nearly 20 times that of \textrm{LC}) which is prohibitive in many practical settings.

Figure~\ref{fig:results2} shows the case of some non-DNA representative texts: natural language and source code. The results stay in the same line as on DNA. We can see that \textrm{Onl} performs noticeably better when $\chi$ is small, while the other constructions (including ours) are much more sensitive to $n$.

\begin{figure}[t]
    \centering
    \begin{minipage}{0.49\textwidth}
        \centering
        \includegraphics[width=\textwidth]{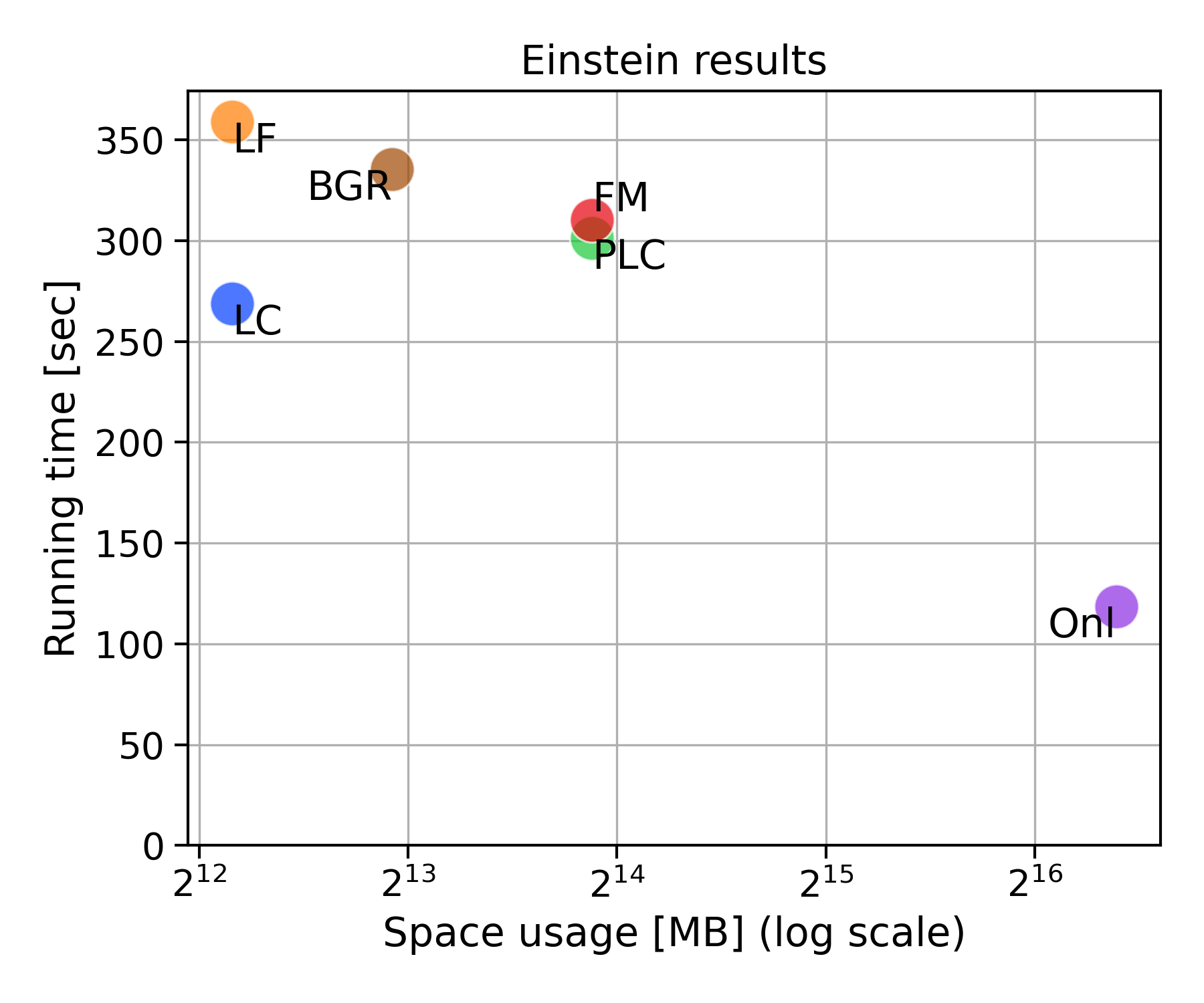} 
    \end{minipage}
    \begin{minipage}{0.49\textwidth}
        \centering
        \includegraphics[width=\textwidth]{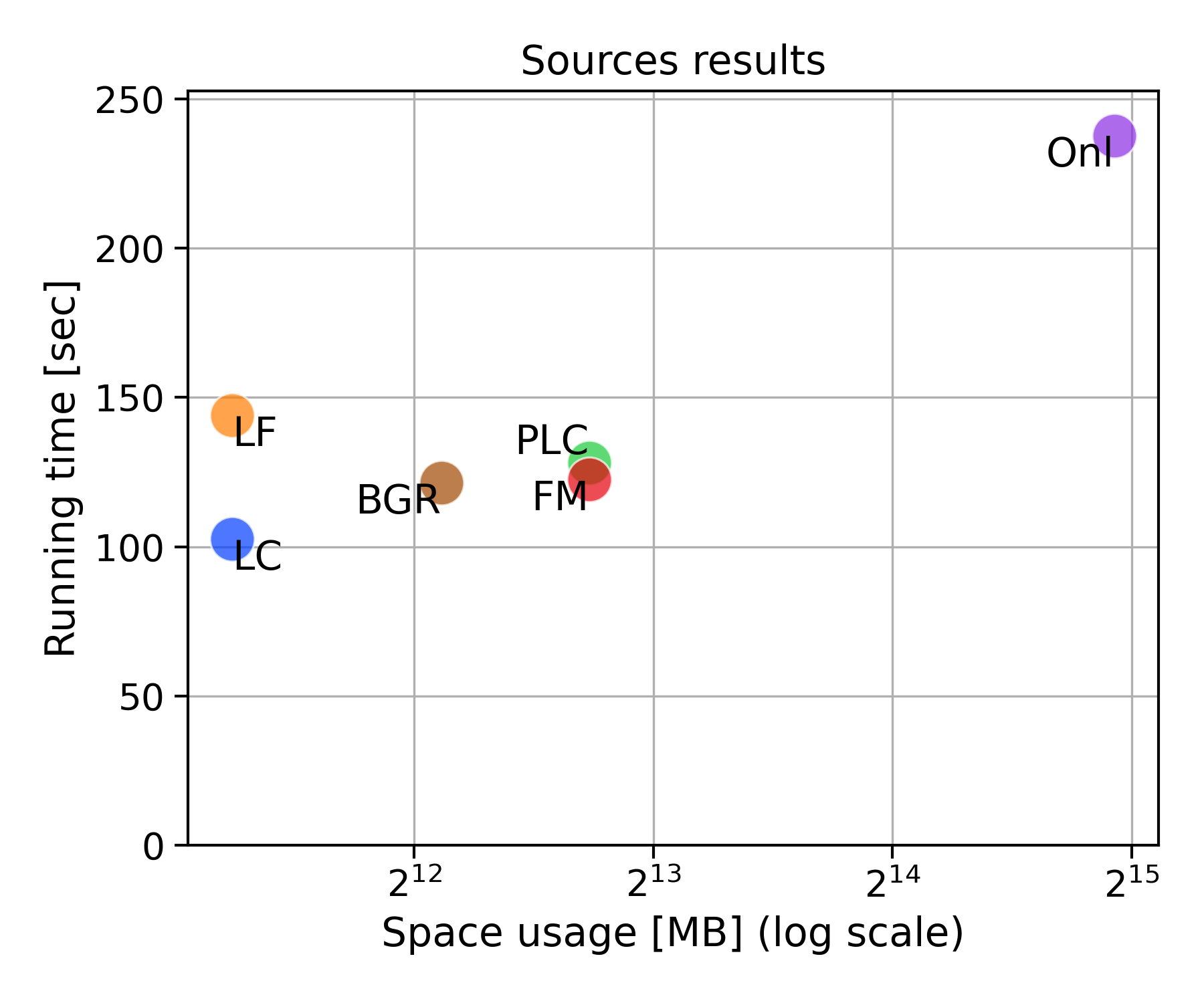} 
    \end{minipage}
     \centering
 \begin{adjustbox}{width=\textwidth}
 \begin{tabular}{|l|r|r||r|r|r|r|r|r||r|r|r|r|r|r|} \hline
 \multicolumn{3}{|c||}{} & \multicolumn{6}{|c||}{\textbf{space (MB)}} & \multicolumn{6}{|c|}{\textbf{time (sec)}} \\ \hline
 Dataset $T$ & \multicolumn{1}{c|}{length $n$} & \multicolumn{1}{c||}{$\chi$} & \multicolumn{1}{|c|}{\textbf{LC}} & \multicolumn{1}{|c|}{\textbf{LF}} & \multicolumn{1}{|c|}{\textbf{PLC}} & \multicolumn{1}{|c|}{\textbf{FM}}  & \multicolumn{1}{|c|}{\textbf{BGR}} & \multicolumn{1}{|c||}{\textbf{Onl}} & \multicolumn{1}{|c|}{\textbf{LC}} & \multicolumn{1}{|c|}{\textbf{LF}} & \multicolumn{1}{|c|}{\textbf{PLC}} & \multicolumn{1}{|c|}{\textbf{FM}} & \multicolumn{1}{|c|}{\textbf{BGR}} & \multicolumn{1}{|c|}{\textbf{Onl}} \\\hline
 Einstein & 467,626,544 & 200,245 & $4571$ & $4571$ & $15077$ & $15077$ & $7776$ & $85854$ & $268.68$ & $359.12$ & $301.32$ & $310.3$ & $335.43$ & $118.46$ \\
 Sources & 210,866,607 & 36,693,523 & $2416$ & $2416$ & $6802$ & $6802$ & $4431$ & $31159$ & $102.54$ & $143.96$ & $128.1$ & $122.57$ & $121.47$ & $237.51$ \\ 
 \hline
 \end{tabular}
 \end{adjustbox}

    \caption{Space usage and running time comparison between all construction algorithms on \textbf{einstein.en.txt} and \textbf{sources} sequences from Pizza\&Chili collection. Logscale in the space coordinate.
    }
    \label{fig:results2}
\end{figure}

\section{Conclusions and Future Work}
We have proposed the first algorithm that is (i) one-pass (over the suffix-array domain), (ii) linear-time, and (iii) implemented, using the same space as the linear-time implementation using least space \cite{TOCSpaper} {\em and} faster than the fastest implementation in that paper. Our algorithm also outperforms our implementation of that of Bonizzoni et al.~\cite{BonizzoniGR2026} in time and space, and shows to be more practical than an online algorithm \cite{fujimaru2026}, in the sense that uses  much less space (by a factor around 20), being faster on some inputs and slower on others. We do not know of other implementations.

As a future line of research, we plan to evaluate if algorithm \textrm{FM} \cite{TOCSpaper} admits the same transformation we made on \textrm{PLC} to obtain \textrm{LC}: computing $\PSV$ and $\NSV$ values \textit{on the fly} and just for run-breaks. That would yield a one-pass version of \textrm{FM} as well. We believe that, with that transformation, \textrm{FM} would be even faster than \textrm{LC} while not using more space.

Being one-pass opens the door to running the construction in streaming mode, where the suffix-array related structures of the reverse text are visited sequentially and not maintained in main memory. A relevant measure for such a streamed construction is the amount of main memory needed. Our algorithm needs only $O(\overline{r})$ main memory space for NSV because it is built only on run breaks, and $O(\sigma)$ for other structures, but it stores an amount of PSV array cells that is proportional to the height of the suffix tree. While small in practice, this height can be up to $\Theta(n)$ in pathologic cases. Reducing this extra space, say to $O(\overline{r})$ as well, would yield a streamed construction whose main memory usage is sensitive to the compressibility of the text.

We also plan to code \textrm{LC} (and eventually the one-pass version of \textrm{FM}) on top of \textrm{PFP} \cite{BoucherGKLMM19}. Cenzato et al.~\cite{TOCSpaper} experiment with running the algorithms over the concatenation of many copies of large DNA sequences, producing huge highly-repetitive files. In such context, the \textrm{PFP} version of the one-pass algorithm performed much better than any of the linear-time algorithms. We believe the \textrm{PFP}-based version of \textrm{LC} can be even faster in that setting.
\bibliographystyle{splncs04}
\bibliography{biblio}

@article{koppl2026,
      title={Smallest suffixient set maintenance in near-real-time}, 
      author={Dominik Köppl and Gregory Kucherov},
      year={2026},
      eprint={2604.27548},
      archivePrefix={arXiv},
      primaryClass={cs.DS},
      url={https://arxiv.org/abs/2604.27548}, 
      journal      = {CoRR},
      volume       = {abs/2604.27548},
      note = "To appear in {\em MFCS'26}"
}

@TECHREPORT { 
        BW94,  
        AUTHOR = "M. Burrows and D. Wheeler",
        TITLE = "A block sorting lossless data compression algorithm",
        INSTITUTION = "Digital Equipment Corporation",
        NUMBER = 124,
        YEAR = 1994
        }

@ARTICLE{
        MNSV09,
        AUTHOR = "V. M{\"a}kinen and G. Navarro and J. Sir{\'e}n and
                  N. V{\"a}lim{\"a}ki",
        TITLE = "Storage and Retrieval of Highly Repetitive Sequence 
                 Collections",
        JOURNAL = {Journal of Computational Biology},
        YEAR = "2010",
        VOLUME = 17,
        NUMBER = 3,
        PAGES = "281--308",
        doi = "10.1089/cmb.2009.0169"
    }

@ARTICLE 
        { CGNtalg25,
          AUTHOR = "D. Cobas and T. Gagie and G. Navarro",
	  TITLE = "Fast and Small Subsampled R-indexes",
	  JOURNAL = "ACM Transactions on Algorithms",
          YEAR = 2025,
	  VOLUME = 22,
	  NUMBER = 1,
	  PAGES = "article 7",
    doi = {10.1145/3750729}
        }

@INPROCEEDINGS{
        NRUspire25.2,
        AUTHOR = "G. Navarro and G. Romana and C. Urbina",
	      TITLE = "Smallest Suffixient Sets as a Repetitiveness Measure",
        BOOKTITLE = {Proc. 32nd International Symposium on String Processing and Information Retrieval (SPIRE)},
        pages="217--232",
        YEAR = "2025",
        doi = "10.1007/978-3-032-05228-5_18",
    }

@article{FerraginaM05,
  author       = {Paolo Ferragina and
                  Giovanni Manzini},
  title        = {Indexing compressed text},
  journal      = {Journal of the ACM},
  volume       = {52},
  number       = {4},
  pages        = {552--581},
  year         = {2005},
  doi          = {10.1145/1082036.1082039}
}

@article{GagieNP20,
  author       = {Travis Gagie and
                  Gonzalo Navarro and
                  Nicola Prezza},
  title        = {Fully Functional Suffix Trees and Optimal Text Searching in BWT-Runs
                  Bounded Space},
  journal      = {Journal of the ACM},
  volume       = {67},
  number       = {1},
  pages        = {2:1--2:54},
  year         = {2020},
  doi          = {10.1145/3375890}
}

@article{TOCSpaper,
  author       = {Davide Cenzato and Lore Depuydt and Travis Gagie and Sung-Hwan Kim and Giovanni Manzini and Francisco Olivares and Nicola Prezza},
  title        = {Suffixient Arrays: a New Efficient Suffix Array Compression Technique}, 
  journal      = {CoRR},
  volume       = {abs/2407.18753},
  year         = {2025},
  doi          = {10.48550/ARXIV.2407.18753}
}

@article{Depuydt24,
  author       = {Lore Depuydt and
                  Travis Gagie and
                  Ben Langmead and
                  Giovanni Manzini and
                  Nicola Prezza},
  title        = {Suffixient Sets},
  journal      = {CoRR},
  volume       = {abs/2312.01359},
  year         = {2023},
  doi          = {10.48550/ARXIV.2312.01359}
}

@inproceedings{CenzatoOP24,
  author       = {Davide Cenzato and
                  Francisco Olivares and
                  Nicola Prezza},
  title        = {On Computing the Smallest Suffixient Set},
  booktitle    = {Proc. 31st International Symposium on String Processing and Information Retrieval (SPIRE)},
  !series       = {Lecture Notes in Computer Science},
  !volume       = {14899},
  pages        = {73--87},
  year         = {2024},
  doi          = {10.1007/978-3-031-72200-4\_6}
}

@article{BerkmanSV93,
    author = {Omer Berkman and Baruch Schieber and Uzi Vishkin},
    title = {Optimal Doubly Logarithmic Parallel Algorithms Based On Finding All Nearest Smaller Values},
    journal = {Journal of Algorithms},
    volume = {14},
    number = {3},
    pages = {344--370},
    year = {1993},
    doi = {10.1006/jagm.1993.1018},
}

@InProceedings{BonizzoniGR2026,
  author =	{Bonizzoni, Paola and Gao, Younan and Riccardi, Brian},
  title =	{{Constructing Suffixient Arrays Revisited}},
  booktitle =	{Proc. 37th Annual Symposium on Combinatorial Pattern Matching (CPM)},
  pages =	{30:1--30:18},
  !series =	{Leibniz International Proceedings in Informatics (LIPIcs)},
  !ISBN =	{978-3-95977-420-8},
  !ISSN =	{1868-8969},
  year =	{2026},
  !volume =	{369},
  !editor =	{Bille, Philip and Prezza, Nicola},
  !publisher =	{Schloss Dagstuhl -- Leibniz-Zentrum f{\"u}r Informatik},
  !address =	{Dagstuhl, Germany},
  !URL =		{https://drops.dagstuhl.de/entities/document/10.4230/LIPIcs.CPM.2026.30},
  !URN =		{urn:nbn:de:0030-drops-259564},
  doi =		{10.4230/LIPIcs.CPM.2026.30},
  !annote =	{Keywords: Suffixient set, suffixient array, right-maximal substring, linear-time algorithm}
}

@misc{fujimaru2026,
      title={Smallest Suffixient Sets: Effectiveness, Resilience, and Calculation}, 
      author={Hiroto Fujimaru and Gonzalo Navarro and Giuseppe Romana and Cristian Urbina},
      year={2026},
      eprint={2506.05638},
      archivePrefix={arXiv},
      primaryClass={cs.FL},
      url={https://arxiv.org/abs/2506.05638}, 
      note = "To appear in {\em Theoretical Computer Science}"
}

@inproceedings{Fuji26,
      title = "Computing Smallest Suffixient Arrays in Almost Sublinear Time",
      author= "Hiroto Fujimaru and Gonzalo Navarro and Francisco Olivares and Jakub Radoszewski and Giuseppe Romana and Cristian Urbina",
      year = 2026,
      booktitle    = {Proc. 33rd International Symposium on String Processing and Information Retrieval (SPIRE)},
      note = "To appear"
}

@misc{cenzato2025testingsuffixientsets,
      title={Testing Suffixient Sets}, 
      author={Davide Cenzato and Francisco Olivares and Nicola Prezza},
      year={2025},
      eprint={2506.08225},
      archivePrefix={arXiv},
      primaryClass={cs.DS},
      url={https://arxiv.org/abs/2506.08225}, 
}

@article{BoucherGKLMM19,
  author       = {Christina Boucher and
                  Travis Gagie and
                  Alan Kuhnle and
                  Ben Langmead and
                  Giovanni Manzini and
                  Taher Mun},
  title        = {Prefix-free parsing for building big {BWT}s},
  journal      = {Algorithms for Molecular Biology},
  volume       = {14},
  number       = {1},
  pages        = {13:1--13:15},
  year         = {2019},
  doi          = {10.1186/S13015-019-0148-5}
}

@InProceedings{sublinear_suffix_array_construction,
author = {Dominik Kempa and Tomasz Kociumaka},
title = {Breaking the {O}(n)-Barrier in the Construction of Compressed Suffix Arrays and Suffix Trees},
booktitle = {Proc. Annual ACM-SIAM Symposium on Discrete Algorithms (SODA)},
pages = {5122-5202},
doi = {10.1137/1.9781611977554.ch187},
!URL = {https://epubs.siam.org/doi/abs/10.1137/1.9781611977554.ch187}
}

\end{document}